\documentclass[1p,final]{elsarticle}

\usepackage{amsfonts,natbib,pslatex,latexsym}

\usepackage{latexsym,amssymb,amsmath,amsthm,amssymb}
\usepackage{color}

\newcommand{\tr}{{\rm Tr}}
\newcommand{\gf}{{\rm GF}}

\newcommand{\C}{{\cal C}}
\newcommand{\D}{{\cal D}}

\newcommand{\Z}{\mathbb {Z}}
\newcommand{\N}{\mathbb {N}}
\newcommand{\M}{\mathbb {M}}
\newcommand{\m}{\mathbb {M}}

\newcommand{\ls}{{\mathbb L}}

\newcommand{\wt}{{\mathrm{wt}}}

\newtheorem{theorem}{Theorem}
\newtheorem{lemma}[theorem]{Lemma}
\newtheorem{remark}[theorem]{Remark}

\newtheorem{open}{Open Problem}

\newtheorem{example}{Example}

\journal{Discrete Mathematics}

\begin{document}

\begin{frontmatter}



\title{Binary Cyclic Codes from Explicit Polynomials over $\gf(2^m)$ \tnoteref{fn1}}
\tnotetext[fn1]{C. Ding's research was supported by
The Hong Kong Research Grants Council, Proj. No. 601013. Z. Zhou's research was supported by
the Natural Science Foundation of China, Proj. No. 61201243, and also The Hong Kong Research
Grants Council, Proj. No. 601013.}


\author[cding]{Cunsheng Ding}
\ead{cding@ust.hk} 
\author[zcz]{Zhengchun Zhou}
\ead{zzc@home.swjtu.edu.cn}

\address[cding]{Department of Computer Science
                                                  and Engineering, The Hong Kong University of Science and Technology,
                                                  Clear Water Bay, Kowloon, Hong Kong, China}
\address[zcz]{School of Mathematics, Southwest Jiaotong University, Chengdu, 610031, China}

\begin{abstract}
Cyclic codes are a subclass of linear codes and have applications in consumer electronics,
data storage systems, and communication systems as they have efficient encoding and
decoding algorithms. In this paper, monomials and trinomials over finite fields with
even characteristic are employed
to construct a number of families of binary cyclic codes. Lower bounds on the minimum
weight of some families of the cyclic codes are developed. The minimum weights of other
families of the codes constructed in this paper are determined. The dimensions of the
codes are flexible.  Some of the codes presented in this paper are optimal or almost optimal
in the sense that they meet some bounds on linear codes. Open problems regarding binary cyclic
codes from monomials and trinomials are also presented.

\end{abstract}


\end{frontmatter}


\section{Introduction}

Let $q$ be a power of a prime $p$.
A linear $[n,k, d]$ code over $\gf(q)$ is a $k$-dimensional subspace of $\gf(q)^n$
with minimum (Hamming) nonzero weight $d$.
A linear $[n,k]$ code $\C$ over the finite field $\gf(q)$ is called {\em cyclic} if
$(c_0,c_1, \cdots, c_{n-1}) \in \C$ implies $(c_{n-1}, c_0, c_1, \cdots, c_{n-2})
\in \C$.
By identifying any vector $(c_0,c_1, \cdots, c_{n-1}) \in \gf(q)^n$
with
$
\sum_{i=0}^{n-1}   c_ix^i  \in \gf(q)[x]/(x^n-1),
$ 
any code $\C$ of length $n$ over $\gf(q)$ corresponds to a subset of $\gf(q)[x]/(x^n-1)$.
The linear code $\C$ is cyclic if and only if the corresponding subset in $\gf(q)[x]/(x^n-1)$
is an ideal of the ring $\gf(q)[x]/(x^n-1)$.
It is well known that every ideal of $\gf(q)[x]/(x^n-1)$ is principal. Let $\C=(g(x))$ be a
cyclic code, where $g(x)$ has the smallest degree and constant term 1. Then $g(x)$ is
called the {\em generator polynomial} and
$h(x)=(x^n-1)/g(x)$ is referred to as the {\em parity-check} polynomial of
$\C$.

The error correcting capability of cyclic codes may not be as good as some other linear
codes in general. However, cyclic codes have wide applications in storage and communication
systems because they have efficient encoding and decoding algorithms
\cite{Chie,Forn,Pran}.
For example, Reed-Solomon codes have found important applications from deep-space
communication to consumer electronics. They are prominently used in consumer
electronics such as CDs, DVDs, Blu-ray Discs, in data transmission technologies
such as DSL \& WiMAX, in broadcast systems such as DVB and ATSC, and in computer
applications such as RAID 6 systems.

Cyclic codes have been studied for decades and a lot of  progress has been made
(see for example, \cite{CLP,CDY05,Dinh,DL06,Feng,GO08,HDLA,HT09,HPbook,JLX11,LF08,LintW,MK93,Mois,PM09,RP10,ZHJ}).  The total number of cyclic codes
over $\gf(q)$ and their constructions are closely related to $q$-cyclotomic cosets
modulo $n$, and thus many topics of number theory. One way of
constructing cyclic codes over $\gf(q)$ with length $n$ is  to use the generator polynomial
\begin{eqnarray}\label{eqn-defseqcode}
\frac{x^n-1}{\gcd(S^n(x), x^n-1)}
\end{eqnarray}
where
$$
S^n(x)=\sum_{i=0}^{n-1} s_i x^i  \in \gf(q)[x]
$$
and $s^{\infty}=(s_i)_{i=0}^{\infty}$ is a sequence of period $n$ over $\gf(q)$.
Throughout this paper, we call the cyclic code $\C_s$ with the generator polynomial
of (\ref{eqn-defseqcode}) the {\em code defined by the sequence} $s^{\infty}$,
and the sequence $s^{\infty}$ the {\em defining sequence} of the cyclic code $\C_s$.

One basic question is whether good cyclic codes can be constructed with
this approach. It turns out that the code $\C_s$ could
be an optimal or almost optimal linear code if the sequence $s^\infty$ is properly
designed \cite{Ding121}.

In this paper, a number of types of monomials and trinomials over $\gf(2^m)$ will be 
employed to construct a number of classes of binary cyclic codes. Lower bounds on 
the minimum weight of some classes of the cyclic codes are developed. The minimum 
weights of some other classes of the codes constructed in this paper are determined. 
The dimensions of the codes
of this paper are flexible.  Some of the codes obtained in this paper are optimal
or almost optimal as they meet certain bounds on linear codes. Several open problems
regarding cyclic codes from monomials and trinomials are also presented in this
paper.

The first motivation of this study is that some of the codes constructed in this paper 
could be optimal or almost optimal. The second motivation is the simplicity of the 
constructions of the cyclic codes that may lead to efficient encoding and decoding 
algorithms.

\section{Preliminaries}

In this section, we present basic notations and results of $q$-cyclotomic cosets, 
highly nonlinear functions, and sequences that will be employed in subsequent sections.

\subsection{Some notations fixed throughout this paper}\label{sec-notations}

Throughout this paper, we adopt the following notations unless otherwise stated:
\begin{itemize}
\item $q=2$, $m$ is a positive integer, $r=q^m$, and $n=r-1$.
\item $\Z_n=\{0,1,2,\cdots, n-1\}$ associated with the integer addition modulo $n$ and
           integer multiplication modulo $n$ operations.
\item $\alpha$ is a generator of $\gf(r)^*$, and 
          $m_a(x)$ is the minimal polynomial of $a \in \gf(r)$ over $\gf(q)$.
\item $\N_q(x)$ is a function defined by $\N_q(i) =0$ if $i \equiv 0 \pmod{q}$ and $\N_q(i) =1$ otherwise, where $i$
          is any nonnegative integer.
\item $\tr(x)$ is the  trace function from $\gf(r)$ to $\gf(q)$.
\item By the Database we mean the collection of the tables of best linear codes known maintained by
         Markus Grassl at http://www.codetables.de/.
\end{itemize}

\subsection{The linear span and minimal polynomial of periodic sequences}\label{sec-sequences}

Let $s^\infty=(s_i)_{i=0}^{\infty}$ be a sequence of period $L$ over $\gf(q)$. 
The polynomial $c(x)= \sum_{i=0}^{\ell}  c_ix^i$ over $\gf(q)$, 
where $c_0=1$, is called a {\em characteristic polynomial} of $s^\infty$ if  
\begin{eqnarray*} 
-c_0s_i=c_1s_{i-1}+c_2s_{i-2}+\cdots +c_ls_{i-\ell} \mbox{ for all }  i \ge \ell. 
\end{eqnarray*} 
The characteristic polynomial with the smallest degree is called the {\em minimal 
polynomial} of $s^\infty$. The degree of the minimal polynomial is referred to as 
the {\em linear span} or {\em linear complexity} of $s^\infty$.
Since we require that the constant term of any characteristic polynomial 
be 1, the minimal polynomial of any periodic sequence $s^{\infty}$ must 
be unique. In addition, any characteristic polynomial must be a multiple 
of the minimal polynomial.    

For periodic sequences, there are a few ways to determine their linear 
span and minimal polynomials. One of them is given in the following 
lemma \cite{LN97}. 

\begin{lemma}\label{lem-ls1} 
Let $s^{\infty}$ be a sequence of period $L$ over $\gf(q)$. 
Define   
$
S^{L}(x)= \sum_{i=0}^{L-1} s_i x^i  \in \gf(q)[x]. 
$ 
Then the minimal polynomial $\m_s(x)$ of $s^{\infty}$ is given by 
      \begin{eqnarray}\label{eqn-base1}  
      \frac{x^{L}-1}{\gcd(x^{L}-1, S^{L}(x))} 
      \end{eqnarray}  
and the linear span $\ls_s$ of $s^{\infty}$ is given by 
      $ 
       L-\deg(\gcd(x^{L}-1, S^{L}(x))). 
      $ 
\end{lemma} 

The other one is given in the following lemma \cite{Antweiler} 

\begin{lemma} \label{lem-ls2} 
Any sequence $s^{\infty}$ over $\gf(q)$ of period $q^m-1$ has a unique expansion of the form  
\begin{equation*}
s_t=\sum_{i=0}^{q^m-2}c_{i}\alpha^{it}, \mbox{ for all } t\ge 0,
\end{equation*}
where $\alpha$ is a generator of $\gf(q^m)^*$ and $c_i \in \gf(q^m)$.
Let the index set be $I=\{i \left.\right| c_i\neq 0\}$, then the minimal polynomial $\m_s(x)$ of $s^{\infty}$ is 
$\m_s(x)=\prod_{i\in I}(1-\alpha^i x),$ 
and the linear span of $s^{\infty}$ is $|I|$.
\end{lemma}

It should be noticed that in some references the reciprocal of $\m_s(x)$ is called the minimal polynomial 
of the sequence $s^\infty$. So Lemma \ref{lem-ls2} is a modified version of the original one in \cite{Antweiler}.

\subsection{The $2$-cyclotomic cosets modulo $2^m-1$}\label{sec-cpsets}

Let $n=2^m-1$. The 2-cyclotomic coset containing $j$ modulo $n$ is defined by
$$
C_j=\{j, 2j, 2^2j, \cdots, 2^{\ell_j-1}j\} \subset \Z_n, 
$$
where $\ell_j$ is the smallest positive integer such that $2^{\ell_j}j \equiv j \pmod{n}$,
and is called the size of $C_j$. It is known that $\ell_j$ divides $m$. The smallest integer
in $C_j$ is called the {\em coset leader} of $C_j$. Let $\Gamma$ denote the set of all
coset leaders and $\Gamma^*=\Gamma\setminus \{0\}$. By definition, we have
$$
\bigcup_{j \in \Gamma} C_j =\Z_n:=\{0,1,2, \cdots, n-1\}
$$
and
$$
C_i \bigcap C_j =\emptyset ~\textrm{for~any}~i\neq j\in \Gamma.
$$

For any integer $j$ with $0\leq j\leq 2^m-1$, the { 2-weight} of $j$, denoted by $\wt(j)$,  is defined to be the number of
nonzero coefficients in its 2-adic expansion:
\begin{align*}
j=j_0+j_1\cdot 2+\cdots+j_{m-1}\cdot 2^{m-1},~~j_i\in \{0,1\}.
\end{align*}

The following lemmas will be useful in the sequel.

\begin{lemma}\cite{SiDing}\label{Lemma_for_coset1}
For any coset leader $j \in \Gamma^*$, $j$ is odd and $1\leq j<2^{n-1}$.
\end{lemma}

\begin{lemma}\label{lemma_odd_number}
For any $j\in \Gamma^*$ with $\ell_j=m$, the number of odd and even integers in the 2-cyclotomic coset $C_j$ 
are equal to $\wt(j)$ and $m-\wt(j)$, respectively.
\end{lemma}

\begin{proof}
By Lemma \ref{Lemma_for_coset1}, $j$ is odd.
Then we can assume that $j=1+2^{i_1}+2^{i_2}+\cdots+2^{i_{k-1}}$ where $k=\wt(j)$ and 
$1\leq i_1<i_2\cdots<i_{k-1}\leq m-1$.
It is easy to check that the odd integers in $C_j$ are given by
\begin{align}\label{eqn_odd_integers}
j,~ j2^{m-i_{k-1}} \bmod{n}, ~j2^{m-i_{k-2}} \bmod{n}, \cdots,  ~j2^{m-i_1} \bmod{n} 
\end{align}
which are pairwise distinct due to $\ell_j=m$.  Thus the number of odd integers in the 2-cyclotomic coset $C_j$ is equal to $k$ and the
number of even ones is equal to $m-k$.
\end{proof}

\begin{lemma}\label{Lemma_equivalent}
All integers $1\leq j\leq 2^m-2$ with $\wt(j)=m-1$ are in the same $2$-cyclotomic coset.
\end{lemma}

\begin{proof}
It is clear that the number of integers $1\leq j\leq 2^m-2$ with $\wt(j)=m-1$ is equal to ${m\choose m-1}=m$. Note that all the integers in the same coset have the weight on the other hand.
The conclusion then follows from the facts that $\wt(2^{m-1}-1)=m-1$ and $\ell_{2^{m-1}-1}=m$.
\end{proof}

\begin{lemma}\label{Lemma_coset}
Let $h$ be an integer with $1\leq h\leq \lceil{m+1\over 2}\rceil$ and $\Gamma_1=\{1\leq j\leq 2^h-1: j ~\textrm{is~odd}\}$.
Then for any $j\in \Gamma_1$,
\begin{itemize}
\item $j$ is the coset leader of $C_j$;

\item $\ell_j=m$ except that $\ell_{2^{m/2}+1}=m/2$ for even $m$.
\end{itemize}
\end{lemma}

\begin{proof}
We first prove the first assertion.  For any $j\in \Gamma_1$, let $\wt(j)=k$ and $j=1+2^{i_1}+2^{i_2}+\cdots+2^{i_{k-1}}$.
where $1\leq i_1<i_2<\cdots<i_{k-1}\leq h-1$. It follows from Lemma \ref{Lemma_for_coset1} that the coset leader must be odd.
By Lemma \ref{lemma_odd_number}, all the odd integers in
the $2$-cyclotomic coset containing $j$ are listed in (\ref{eqn_odd_integers}) in which
the least one is exactly $j$ due to $i_t\leq m/2$ for all $1\leq t\leq k-1$. This finishes the proof of the first
assertion.

We now prove the second one. Note that for each $j\in \Gamma_1$,  $\ell_j| m$
and $j$ is divisible  by  $(2^m-1)/(2^{\ell_j}-1)$. When $m$ is odd, if $\ell_j<m$, then $\ell_j\leq {m/3}$ and thus
$(2^m-1)/(2^{\ell_j}-1)>2^{2m/3}$ which means that $j> 2^{2m/3}$. This is impossible since $j<2^{(m+1)/2}$. Thus
$\ell_j=m$ for odd $m$.
Similarly, when  $m$ is even, if $\ell_j< {m}$, then $\ell_j\leq {m/2}$ and
$(2^m-1)/(2^{\ell_j}-1)>2^{m-\ell_j}$. It is easy to check that $j\in \Gamma_1$ is divisible by $(2^m-1)/(2^{\ell_j}-1)$ if and only
if $j=2^{m/2}+1$ and $\ell_j=m/2$.
\end{proof}

\subsection{PN and APN functions}

A polynomial $f(x)$ over $\gf(r)$ is called {\em almost perfect nonlinear (APN)} if
$$
\max_{a \in \gf(r)^*} \max_{b \in \gf(r)} |\{x \in \gf(r): f(x+a)-f(x)=b\}| =2,
$$
and is referred to as
{\em perfect
nonlinear or planar} if
$$
\max_{a \in \gf(r)^*} \max_{b \in \gf(r)} |\{x \in \gf(r): f(x+a)-f(x)=b\}| =1.
$$
In subsequent sections, we need the notion of PN and APN functions. 

\section{Codes defined by polynomials over finite fields $\gf(r)$}\label{sec-sequence}

\subsection{A generic construction of cyclic codes with polynomials}\label{sec-gconstruct}

Given any polynomial $f(x)$ over $\gf(r)$, we define its associated sequence
$s^\infty$ by
\begin{eqnarray}\label{eqn-sequence}
s_i=\tr\left(f\left(\alpha^i+1\right)\right)
\end{eqnarray}
for all $i \ge 0$.

The objective of this paper is to consider the codes $\C_s$ defined by some monomials and
trinomials over $\gf(2^m)$.

\subsection{How to choose the polynomial $f(x)$}\label{sec-howto}

Regarding the generic construction of Section \ref{sec-gconstruct}, the
following two questions are essential.
\begin{itemize}
\item Is it possible to construct optimal cyclic codes meeting some bound on
          parameters of linear codes or cyclic codes with good parameters?
\item If the answer to the question above is positive, how should we select the polynomial $f(x)$
          over $\gf(r)$?
\end{itemize}

It will be demonstrated in the sequel that the answer to the first question above is indeed
positive. However, it seems hard to answer the second question.

Any method of constructing an $[n, k]$ cyclic code over $\gf(q)$ becomes the selection
of a divisor $g(x)$ over $\gf(q)$ of $x^n-1$ with degree $n-k$, which is employed as the
generator polynomial of the cyclic code. The minimum weight $d$ and other parameters
of this cyclic code are determined by the generator polynomial $g(x)$.

Suppose that an optimal $[n, k]$ cyclic code over $\gf(q)$ exists. The question is how to
find out the divisor $g(x)$ of $x^n-1$ that generates the optimal cyclic code. Note that
$x^n-1$ may have many divisors of small degrees. If the construction method is not well
designed, optimal cyclic codes cannot be produced even if they exist.

The construction of Section \ref{sec-gconstruct} may produce cyclic codes with bad
parameters. For example, let $(q,m)=(2,6)$, let $\alpha$ be the generator of $\gf(2^6)$
with $\alpha^6 + \alpha^4
+ \alpha^3 + \alpha + 1=0$,  and let $f(x)=x^e$. When $e \in
\{7, 14, 28, 35,  49, 56\}$, the binary code $\C_s$ defined by the monomial $f(x)$
has parameters $[63, 45, 3]$. These codes are very bad as there are binary linear
codes with parameters $[63, 45, 8]$ and binary cyclic codes with parameters $[63, 57, 3]$.

On the other hand, the construction of Section \ref{sec-gconstruct} may produce
optimal cyclic codes. For example, let $(q,m)=(2,6)$ and let $f(x)=x^e$. When
$$e \in
\{1, 2, 4, 5, 8, 10, 16, 17, 20, 32, 34, 40\},$$
the binary code $\C_s$ defined by the
monomial $f(x)$ has parameters $[63, 57, 3]$ and should be equivalent to the binary
Hamming code with the same parameters. These cyclic codes are optimal with respect
to the sphere packing bound.

Hence, a monomial may give good or bad cyclic codes within the framework of the
construction of Section \ref{sec-gconstruct}. Now the question is how to choose
a monomial $f(x)$ over $\gf(r)$ so that the cyclic code $\C_s$ defined by $f(x)$
has good parameters.

In this paper, we employ monomial and tronomials $f(x)$ over $\gf(r)$ that are
either permutations on $\gf(r)$ or such that $|f(\gf(r))|$ is very close to $r$. Most of the
monomials and trinomials $f(x)$ employed in this paper are either almost perfect
nonlinear or planar functions on $\gf(r)$.

It is unnecessary to require that $f(x)$ be highly nonlinear, to obtain cyclic codes $\C_s$
with good parameters. Both linear and highly nonlinear polynomials $f(x)$ could give
optimal cyclic codes $\C_s$ when they are plugged into the generic construction of
Section \ref{sec-gconstruct}.


\section{Binary cyclic codes from the permutation monomial $f(x)=x^{2^{t}+3}$}\label{sec-Welch}

In this section we study the code $\C_s$ defined by the permutation monomial $f(x)=x^{2^{t}+3}$ over $\gf(2^{2t+1})$.
Before doing this, we need to prove the following lemma.

\begin{lemma}\label{lem-Welch}
Let $m =2t+1 \ge 7$.
Let $s^{\infty}$ be the sequence of (\ref{eqn-sequence}), where $f(x)=x^{2^{t}+3}$.
Then the linear span $\ls_s$ of $s^{\infty}$ is equal to $5m+1$ and the minimal polynomial $\m_s(x)$
of  $s^{\infty}$ is given by
\begin{eqnarray}\label{eqn-Welch}
\m_s(x)=  (x-1) m_{\alpha^{-1}}(x) m_{\alpha^{-3}}(x) m_{\alpha^{-(2^t+1)}}(x)m_{\alpha^{-(2^t+2)}}(x) m_{\alpha^{-(2^t+3)}}(x).
\end{eqnarray}
\end{lemma}

\begin{proof}
By definition, we have
\begin{eqnarray}\label{eqn-Welch2}
s_i &=& \tr\left((\alpha^i+1)^{2^t+2+1}\right) \nonumber \\
&=& \tr\left( (\alpha^i)^{2^t+3} + (\alpha^i)^{2^t+2} + (\alpha^i)^{2^t+1} +(\alpha^i)^{3} +\alpha^i +1    \right) \nonumber \\
&=& \sum_{j=0}^{m-1} (\alpha^i)^{(2^t+3)2^j} + \sum_{j=0}^{m-1}  (\alpha^i)^{(2^{t-1}+1)2^j} + \sum_{j=0}^{m-1}  (\alpha^i)^{(2^t+1)2^j} \nonumber\\
&& + \sum_{j=0}^{m-1}   (\alpha^i)^{3\times2^j} +  \sum_{j=0}^{m-1}   (\alpha^i)^{2^j} + 1.
\end{eqnarray}

By Lemma \ref{Lemma_coset},  the following $2$-cyclotomic cosets are pairwise disjoint and have size $m$:
\begin{eqnarray}\label{eqn-fourcosets}
C_1, \ C_3, \ C_{2^t+1}, \ C_{2^{t-1}+1}.
\end{eqnarray}

It is clear that the $2$-cyclotomic coset $C_{2^t+3}$ are
disjoint with all the cosets in (\ref{eqn-fourcosets}).
We now prove  $C_{2^t+3}$ has size $m$. It is sufficient to prove that
$
\gcd:=\gcd(2^{2t+1}-1, 2^t+3)=1.
$
The conclusion is true for all $1 \le t \le 4$. So we consider only the case that $t \ge 5$.

Note that
$
2^{2t+1}-1=(2^{t+1}-6)(2^{t}+3) + 17.
$ 
We have $\gcd=\gcd(2^{t}+3, 17)$. Since
$ 
2^{t}+3=2^{t-4}(2^{4}+1)-(2^{t-4}-3),
$ 
we obtain that $\gcd=\gcd(2^{t-4}-3, 2^3-1)$.
Let $t_1=\lfloor t/4\rfloor$. Using the Euclidean division recursively, one gets
\begin{eqnarray*}
\gcd &=& \gcd(2^{t-1-4t_1}+3 (-1)^{t_1}, 2^3-1)\\
&=& \left\{ \begin{array}{ll}
       \gcd(2^0+(-1)^{t_1}3, 17)=1 & \mbox{if } t \equiv 0 \pmod{4}, \\
       \gcd(2^1+(-1)^{t_1}3, 17)=1 & \mbox{if } t \equiv 1 \pmod{4}, \\
       \gcd(2^2+(-1)^{t_1}3, 17)=1 & \mbox{if } t \equiv 2 \pmod{4}, \\
       \gcd(2^3+(-1)^{t_1}3, 17)=1 & \mbox{if } t \equiv 3 \pmod{4}.
       \end{array}
       \right.
\end{eqnarray*}
Therefore $\ell_{2^t+3}=|C_{2^t+3}|=\ell_{-(2^t+3)}=m$.

The desired conclusions on the linear span and the minimal polynomial $\m_s(x)$ then follow from Lemma \ref{lem-ls2},
(\ref{eqn-Welch2}) and the conclusions on the five cyclotomic cosets and their sizes.
\end{proof}

The following theorem provides information on the code $\C_{s}$.

\begin{theorem}\label{thm-Welch}
Let $m \ge 7$ be odd.
The binary code $\C_{s}$ defined by the sequence of Lemma \ref{lem-Welch} has parameters
$[2^m-1, 2^{m}-2-5m, d]$ and  generator polynomial $\m_s(x)$ of (\ref{eqn-Welch}), where $d \ge 8$.
\end{theorem}

\begin{proof}
The dimension of $\C_{s}$ follows from Lemma \ref{lem-Welch} and the definition of the
code $\C_s$.  We need to prove the conclusion on the minimum distance $d$ of $\C_{s}$.
To this end, let ${\D_{s}}$ denote the cyclic code with generator polynomial
$m_{\alpha^{-1}}(x) m_{\alpha^{-(2^{j}+1)}}(x) m_{\alpha^{-(2^{2j}+1)}}(x)$, and let $\overline{\D_{s}}$ 
be the even-weight subcode of ${\D_{s}}$. 
Then ${\D_{s}}$ is a triple-error-correcting code for any $j$ with $\gcd(j,m)=1$ \cite{Kasa}.
This means that the minimal distance of ${\D_{s}}$ is equal to 7 and that of $\overline{\D_{s}}$ is 8. 
Take $j=t+1$, then ${\D_{s}}$ has generator polynomial
$m_{\alpha^{-1}}(x) m_{\alpha^{-3}}(x) m_{\alpha^{-(2^{t}+1)}}(x)$ and $\C_{s}$ is a subcode of   $\overline{\D_{s}}$.
The conclusion then follows from the fact that $x-1$ is a factor of $\m_s(x)$.
\end{proof}

\begin{example}
Let $m=5$ and $\alpha$ be a generator of $\gf(2^m)^*$ with $\alpha^5 + \alpha^2 + 1=0$. Then
the generator polynomial of the code $\C_s$ is
$
\m_s(x)=x^{16} + x^{15} + x^{13} + x^{12} + x^8 + x^6 + x^3 + 1,
$
and $\C_s$ is a $[31, 15, 8]$ binary cyclic code. Its dual is a $[31,16,7]$ cyclic code. Both codes are optimal
according to the Database.
\end{example}

\begin{example}\label{ex-welch2}
Let $m=7$ and $\alpha$ be a generator of $\gf(2^m)^*$ with $\alpha^7 + \alpha + 1=0$. Then
the generator polynomial of the code $\C_s$ is
$ 
\m_s(x) =  x^{36} + x^{34} + x^{33} + x^{32} + x^{29} + x^{28} + x^{27} +
  x^{26} + x^{25} + 
    x^{24} +  x^{21} + x^{12} + x^{11} + x^9 + x^7 + x^6 + x^5 + x^3 + x + 1
$ 
and $\C_s$ is a $[127, 91, 8]$ binary cyclic code.
\end{example}

It can be seen from Example \ref{ex-welch2}
that the bound on the minimal distance of $\C_s$ in
Theorem \ref{thm-Welch} is tight in certain cases.

\section{Binary cyclic codes from the permutation monomial $f(x)=x^{2^h-1}$}\label{sec-2hminus1}

Consider monomials over $\gf(2^m)$ of the form $f(x)=x^{2^h-1}$, where $h$ is a positive integer with $1\leq h\leq {\lceil {m\over 2} \rceil}$.

In this section, we deal with the binary code $\C_s$ defined by the sequence $s^{\infty}$
of (\ref{eqn-sequence}), where $f(x)=x^{2^h-1}$.

We need to do some preparations before presenting and proving the main results
of this section. Let $t$ be a positive integer. We define $T=2^t-1$. For any
odd $a \in \{1,2,3,\cdots,T\}$, define
\begin{equation*}
\epsilon_a^{(t)} =\left\{
\begin{array}{ll}
1, &\textrm{if~} a=2^h-1\\
\left\lceil {\log_2{T\over a}}\right\rceil \bmod 2,&\textrm{if~}   1\leq a<2^h-1.
\end{array} \right.\ \
\end{equation*}
and
\begin{equation}\label{eqn-def-kappa}
\kappa_a^{(t)} = \epsilon_a^{(t)}~\bmod 2.
\end{equation}
Let
$$
B_a^{(t)} =\left\{2^ia: i =0,1,2, \cdots, \epsilon_a^{(t)} -1 \right\}.
$$
Then it can be verified that
$$
\bigcup_{1 \le 2j+1 \le T} B_{2j+1}^{(t)} =\{1,2,3,\cdots, T\}
$$
and
$$
B_a^{(t)} \cap B_b^{(t)} = \emptyset
$$
for any pair of distinct odd numbers $a$ and $b$ in  $\{1,2,3,\cdots, T\}$.

The following lemma follows directly from the definitions of $\epsilon_a^{(t)}$
and $B_a^{(t)}$.

\begin{lemma}\label{lem-f263}
Let $a$ be an odd integer in $\{0,1,2. \cdots, T\}$. Then
\begin{eqnarray*}
& & B_a^{(t+1)} = B_a^{(t)} \cup \{a 2^{\epsilon_a^{(t)}}\} \mbox{ if } 1 \le a \le 2^t-1, \\
& & B_a^{(t+1)} = \{a\} \mbox{ if } 2^t+1 \le a \le 2^{t+1}-1, \\
& & \epsilon_a^{(t+1)} = \epsilon_a^{(t)} +1  \mbox{ if } 1 \le a \le 2^t-1, \\
& & \epsilon_a^{(t+1)} = 1 \mbox{ if } 2^t+1 \le a \le 2^{t+1}-1.
\end{eqnarray*}
\end{lemma}

\begin{lemma}\label{lem-f264}
Let $N_t$ denote the total number of odd  $\epsilon_a^{(t)}$ when $a$ ranges over all
odd numbers in the set $\{1,2,\cdots, T\}$. Then $N_1=1$ and
$$
N_t = \frac{2^t+(-1)^{t-1}}{3}
$$
for all $t \ge 2$.
\end{lemma}

\begin{proof}
It is easily checked that $N_2=1$, $N_3=3$ and $N_4=5$.
It follows from Lemma \ref{lem-f263} that
$$
N_t= 2^{t-2} + (2^{t-2} -N_{t-1}).
$$
Hence
$$
N_t - 2^{t-2}  = 2^{t-3} - (N_{t-1}-2^{t-3})= 3\times 2^{t-4} + (N_{t-2}-2^{t-4}).
$$
With the recurcive application of this recurrence formula, one obtains the desired
formula for $N_t$.
\end{proof}

\begin{lemma}\label{lem-22mm1}
Let $s^{\infty}$ be the sequence of (\ref{eqn-sequence}), where $f(x)=x^{2^h-1}$, $2\leq h\leq {\lceil {m\over 2} \rceil}$. Then the linear span $\ls_s$ of $s^{\infty}$ is given by
\begin{eqnarray}\label{eqn-22m0}
\ls_s =\left\{ \begin{array}{l}
                   \frac{m(2^h+(-1)^{h-1})}{3},~ \mbox{ if $m$ is even} \\
                   \frac{m(2^h+(-1)^{h-1}) +3}{3},~ \mbox{ if $m$ is odd.}
\end{array}
\right.
\end{eqnarray}
We have then
\begin{equation}\label{eqn-2m31}
\m_s(x) =
 (x-1)^{\N_2(m)} \prod_{1 \le 2j+1 \le 2^h-1 \atop \kappa_{2j+1}^{(h)} =1} m_{\alpha^{-(2j+1)}}(x).
\end{equation}
\end{lemma}

\begin{proof}
We have
\begin{eqnarray}\label{eqn-22m11}
\tr(f(x+1))
&= & \tr\left(  (x+1)^{\sum_{i=0}^{h-1} 2^{i}}    \right)
= \tr\left( \prod_{i=0}^{h-1} \left(x^{2^{i}}+1\right)     \right)
= \tr\left( \sum_{i=0}^{2^h-1} x^{i}     \right) \nonumber \\
&=&\tr(1) + \tr\left( \sum_{i=1}^{2^h-1} x^{i}     \right)
= \tr(1) + \tr\left( \sum_{1 \le 2i+1 \le 2^h-1 \atop \kappa_{2i+1}^{(h)} =1} x^{2i+1}     \right)
\end{eqnarray}
where the last equality follows from Lemma \ref{Lemma_coset}.

By definition, the sequence of (\ref{eqn-sequence}) is given by $s_t=\tr(f(\alpha^t+1))$ for all $t \ge 0$.
The desired conclusions on the linear span and the minimal polynomial $\m_s(x)$ then follow from Lemmas
\ref{Lemma_coset}, \ref{lem-f264} and Equation
(\ref{eqn-22m11}).
\end{proof}

The following theorem provides information on the code $\C_{s}$.

\begin{theorem}\label{thm-38}
Let $h \ge 2$.
The binary code $\C_{s}$ defined by the binary sequence of Lemma \ref{lem-22mm1} has parameters
$[2^m-1, 2^{m}-1-\ls_s, d]$ and generator polynomial $\m_s(x)$ of  (\ref{eqn-2m31}),
 where $\ls_s$ is  given in (\ref{eqn-22m0}) and
 \begin{eqnarray*}
 d \ge \left\{ \begin{array}{l}
                     2^{h-2}+2 \mbox{ if $m$ is odd and $h>2$} \\
                     2^{h-2}+1.
                     \end{array}
 \right.
 \end{eqnarray*}
\end{theorem}

\begin{proof}
The dimension of $\C_{s}$ follows from Lemma \ref{lem-22mm1} and the definition of the
code $\C_s$. We now derive the lower bounds on the minimum weight $d$ of the code. It is
well known that the codes generated by $\m_s(x)$ and its reciprocal have the same weight
distribution. It follows from Lemmas \ref{lem-22mm1}  and \ref{lem-f263} that the reciprocal
of $\m_s(x)$ has zeros $\alpha^{2j+1}$ for all $j$ in $\{2^{h-2}, 2^{h-2}+1, \cdots, 2^{h-1}-1\}$.
By the Hartman-Tzeng bound, we have $d \ge 2^{h-2}+1$. If $m$ is odd, $\C_s$ is an
even-weight code. In this case, $d \ge 2^{h-2}+2$.
\end{proof}

\begin{example}
Let $(m,h)=(7,2)$ and $\alpha$ be a generator of $\gf(2^m)^*$ with $\alpha^7 + \alpha + 1=0$. Then
the generator polynomial of the code $\C_s$ is
$
\m_s(x) = x^8 + x^6 + x^5 + x^4 + x^3 + x^2 + x + 1,
$
and $\C_s$ is a $[127, 119, 4]$ binary cyclic code and optimal according to the Database.
\end{example}

\begin{example}
Let $(m, h)=(7,3)$ and $\alpha$ be a generator of $\gf(2^m)^*$ with $\alpha^7 + \alpha + 1=0$. Then
the generator polynomial of the code $\C_s$ is
$
\m_s(x) = x^{22} + x^{21} + x^{20} + x^{18} + x^{17} + x^{16} + x^{14} +
 x^{13} + x^8 + x^7 + x^6 + x^5 + x^4 + 1
$ 
and $\C_s$ is a $[127, 105, d]$ binary cyclic code, where $4 \le d \le 8$.
\end{example}

\begin{remark}
The code $\C_s$ of Theorem \ref{thm-38} may be bad when $\gcd(h, m) \ne 1$. In this case the monomial
$f(x)=x^{2^h-1}$ is not a permutation of $\gf(2^m)$. For example, when $(m, h)=(6,3)$, $\C_s$ is a
$[63, 45, 3]$ binary cyclic code, while the best known linear code in the Database has parameters $[63, 45, 8]$.  Hence,
we are interested in this code only for the case that $\gcd(h, m)=1$, which guarantees that $f(x)=x^{2^h-1}$ is a
permutation of $\gf(2^m)$.
\end{remark}

\section{Binary cyclic codes from the permutation monomial $f(x)=x^e$, $e=2^{(m-1)/2}+2^{(m-1)/4}-1$ and $m \equiv 1 \pmod{4}$}\label{sec-1Niho}

Let $f(x)=x^e$, where $e=2^{(m-1)/2}+2^{(m-1)/4}-1$ and $m \equiv 1 \pmod{4}$.
It can be proved that $f(x)$ is a permutation of $\gf(r)$.
Define $h=(m-1)/4$. We have then
\begin{eqnarray}\label{eqn-1Niho}
\tr(f(x+1))
&=& \tr\left( (x^{2^{2h}}+1) (x+1) ^{\sum_{i=0}^{h-1} 2^{i}}    \right)
= \tr\left( (x^{2^{2h}}+1) \prod_{i=0}^{h-1} \left(x^{2^{i}}+1\right)     \right) \nonumber \\
&=& \tr\left( (x^{2^{2h}}+1) \sum_{i=0}^{2^h-1} x^{i}     \right)
= 1+ \tr\left(\sum_{i=0}^{2^h-1} x^{i+2^{2h}}  + \sum_{i=1}^{2^h-1} x^{i}     \right).
\end{eqnarray}

The sequence $s^{\infty}$ of (\ref{eqn-sequence}) defined by the the monomial $f(x)=x^e$ is then
given by
\begin{eqnarray}\label{eqn-1Nihoseq}
s_t= 1+ \tr\left(\sum_{i=0}^{2^h-1} (\alpha^t)^{i+2^{2h}}  + \sum_{i=1}^{2^h-1} (\alpha^t)^{i}     \right)
\end{eqnarray}
for all $t \ge 0$, where $\alpha$ is a generator of $\gf(2^m)^*$.
In this section, we deal with the code $\C_s$ defined by the sequence $s^{\infty}$ of
(\ref{eqn-1Nihoseq}). To this end, we need to prove a number of auxiliary results on
$2$-cyclotomic cosets.

We define the following two sets for convenience:
\begin{eqnarray*}
A=\{0,1,2, \cdots, 2^h-1\}, \ B=2^{2h}+A=\{i+2^{2h}: i \in A\}.
\end{eqnarray*}

\begin{lemma}\label{lem-1Nf261}
For any $j \in B$, the size $\ell_j=|C_j|=m$.
\end{lemma}

\begin{proof}
Let $j = i + 2^{2h}$, where $i \in A$. For any $u$ with $1 \le u \le m-1$, define
\begin{eqnarray*}
\Delta_1(j, u) = j(2^{u}-1)=(i+2^{2h}) (2^{u}-1), \ \
\Delta_2(j, u) = j(2^{m-u}-1)=(i+2^{2h}) (2^{m-u}-1).
\end{eqnarray*}
If $\ell_j <m$, there would be an integer $1 \le u \le m-1$ such that $\Delta_t(j,u) \equiv 0 \pmod{n}$
for all $t \in \{1,2\}$.

Note that $1 \le u \le m-1$. We have that $\Delta_1(j, u) \ne 0$ and  $\Delta_2(j, u) \ne 0$.
When $u \le m-2h-1$, we have
$$
2^h \le \Delta_1(j, u) \le (2^{2h}+2^h-1)(2^{m-2h-1}-1) <n.
$$
In this case, $\Delta_1(j,u) \not\equiv 0 \pmod{n}$.

When $u \ge m-2h$, we have $m-u \le 2h$ and
$$
2^h \le \Delta_2(j, u) \le (2^{2h}+2^h-1)(2^{2h}-1) <n.
$$
In this case, $\Delta_2(j,u) \not\equiv 0 \pmod{n}$.

Combining the conclusions of the two cases above completes the proof.
\end{proof}

\begin{lemma}\label{lem-1Nf262}
For any pair of distinct $i$ and $j$ in $B$,
$C_i \cap C_j = \emptyset$, i.e., they cannot be in the same $2$-cyclotomic
coset modulo $n$.
\end{lemma}

\begin{proof}
Let $i=i_1+2^{2h}$ and $j=j_1+2^{2h}$, where $i_1 \in A$ and $j_1 \in A$. Define
\begin{eqnarray*}
& & \Delta_1(i, j, u) = i2^{u}-j=  (i_1+2^{2h})2^u-  (j_1+2^{2h}), \\
& & \Delta_2(i, j, u) = j2^{m-u}-i= (j_1+2^{2h})2^{m-u}-  (i_1+2^{2h}).
\end{eqnarray*}
If $C_i=C_j$, there would be an integer $1 \le u \le m-1$ such that $\Delta_t(i,j,u) \equiv 0 \pmod{n}$
for all $t \in \{1,2\}$.

We first prove that $\Delta_1(i, j, u) \ne 0$. When $u=0$,  $\Delta_1(i, j, u)=i_1-j_1 \ne 0$. When
$1 \le u \le m-1$, we have
$$
\Delta_1(i, j, u) \ge 2i_1 + 2^{2h+1}-2^{2h}-j_1 >0.
$$

Since $1 \le u \le m-1$, one can similarly prove that $\Delta_2(i, j, u) >0$.

When $u \le m-2h-1$, we have
$$
-n < -2^{2h} \le \Delta_1(i,j, u) \le (2^{2h}+2^h-1)(2^{m-2h-1}-1) <n.
$$
In this case, $\Delta_1(i,j,u) \not\equiv 0 \pmod{n}$.

When $u \ge m-2h$, we have $m-u \le 2h$ and
$$
0< \Delta_2(i, j, u) \le (2^{2h}+2^h-1)2^{2h}-i_1-2^h <n.
$$
In this case, $\Delta_2(i,j,u) \not\equiv 0 \pmod{n}$.

Combining the conclusions of the two cases above completes the proof.
\end{proof}

\begin{lemma}\label{lem-feb281}
For any $i+2^{2h} \in B$ and odd $j \in A$,
\begin{eqnarray}
C_{i+2^{2h}} \cap C_j = \left\{ \begin{array}{l}
                             C_j \mbox{ if } (i,j)=(0,1) \\
                             \emptyset \mbox{ otherwise.}
\end{array}
\right.
\end{eqnarray}
\end{lemma}

\begin{proof}
Define
\begin{eqnarray*}
\Delta_1(i, j, u) = j2^{u}-(i+2^{2h}),  \ \
\Delta_2(i, j, u) =  (i+2^{2h})2^{m-u}-  j.
\end{eqnarray*}
Suppose $C_{i+2^{2h}}=C_j$, there would be an integer $0 \le u \le m-1$ such that $\Delta_t(i,j,u) \equiv 0 \pmod{n}$
for all $t \in \{1,2\}$.

If $u=2h$, then
\begin{eqnarray*}
0 \equiv \Delta_1(i, j, u) & \equiv & 2^{2h+1}(j2^{2h}-(i+2^{2h})) \pmod{n} \\
 & \equiv & j2^{m}-i2^{2h+1}-2^{m} \pmod{n} \\
 & \equiv &  j -1 -i 2^{2h+1} \pmod{n} \\
 & = &  j -1 -i 2^{2h+1}.
\end{eqnarray*}
Whence, the only solution of $\Delta_1(i,j,2h) \equiv 0 \pmod{n}$ is $(i,j)=(0,1)$.

We now consider the case that $0 \le u <2h$. We claim that $\Delta_1(i, j, u) \ne 0$.
Suppose on the contrary that  $\Delta_1(i, j, u) = 0$. We would then have
$ 
j 2^u -i - 2^{2h} =0.
$ 
Because $u<2h$ and $j$ is odd, there is an odd $i_1$ such that $i=2^u i_1$. It then
follows from $i < 2^h$ that $u <h$. We obtain then
$$
j=i_1+2^{2h-u}>i_1 + 2^{h}>2^h-1.
$$
This is contrary to the assumption that $j \in A$. This proves that $\Delta_1(i, j, u) \ne 0$.

Finally, we deal with the case that $2h+1 \le u <4h=m-1$. We prove that $\Delta_2(i, j, u) \not\equiv 0 \pmod{n}$
in this case.  Since $j$ is odd, $\Delta_2(i, j, u) \ne 0$. We have also
\begin{eqnarray*}
\Delta_2(i, j, u)
= i2^{m-u}+2^{m+2h-u} -j
\le  (2^h-1)2^{m-u} + 2^{m-1} -j
\le  2^{m-(h-1)}+2^{m-1}-j
< n.
\end{eqnarray*}
Clearly, $\Delta_2(i, j, u) >-j >-n$. Hence in this case we have $\Delta_2(i, j, u) \not\equiv 0 \pmod{n}$.

Summarizing the conclusions above proves this lemma.
\end{proof}

\begin{lemma}\label{lem-1N2m1}
Let $m \ge 9$ be odd.
Let $s^{\infty}$ be the sequence of (\ref{eqn-1Nihoseq}). Then the linear span $\ls_s$ of $s^{\infty}$ is given by
\begin{eqnarray}\label{eqn-1N2m0}
\ls_s =\left\{ \begin{array}{l}
                   \frac{m\left(2^{(m+7)/4}+(-1)^{(m-5)/4}\right) +3}{3}, ~\mbox{ if $m \equiv 1 \pmod{8}$} \\
                   \frac{m\left(2^{(m+7)/4}+(-1)^{(m-5)/4}-6\right) +3}{3}, ~\mbox{ if $m \equiv 5 \pmod{8}$.}
\end{array}
\right.
\end{eqnarray}
We have also
\begin{equation*}\label{eqn-1N2m21}
\m_s(x) = (x-1) \prod_{i=0}^{2^{\frac{m-1}{4}}-1} m_{\alpha^{-i-2^{\frac{m-1}{2}} }}(x)
\prod_{1 \le 2j+1 \le 2^{\frac{m-1}{4}}-1 \atop \kappa_{2j+1}^{((m-1)/4)} =1} m_{\alpha^{-2j-1}}(x)
\end{equation*}
if $m \equiv 1 \pmod{8}$; and
\begin{equation*}\label{eqn-1N2m31}
\m_s(x) =(x-1) \prod_{i=1}^{2^{\frac{m-1}{4}}-1} m_{\alpha^{-i-2^{\frac{m-1}{2}} }}(x)
\prod_{3 \le 2j+1 \le 2^{\frac{m-1}{4}}-1 \atop \kappa_{2j+1}^{((m-1)/4)} =1} m_{\alpha^{-2j-1}}(x)
\end{equation*}
if $m \equiv 5 \pmod{8}$,
where $\kappa_{2j+1}^{(h)}$ was
defined in Section \ref{sec-2hminus1}.
\end{lemma}

\begin{proof}
By Lemma \ref{lem-1Nf262}, the monomials in the  function
\begin{equation}\label{eqn-feb28111}
\tr\left(\sum_{i=0}^{2^h-1} x^{i+2^{2h}} \right)
\end{equation}
will not cancel each other. Lemmas \ref{lem-f264} and \ref{lem-22mm1} say that after cancellation, we have
\begin{equation}\label{eqn-feb28121}
\tr\left(\sum_{i=1}^{2^h-1} x^{i} \right) =
\tr\left(\sum_{1 \le 2j+1 \le 2^h-1 \atop \kappa_{2j+1}^{(h)}=1} x^{2j+1} \right).
\end{equation}

By Lemma \ref{lem-feb281}, the monomials in the function of (\ref{eqn-feb28111}) will not cancel the monomials in
the function in the right-hand side of (\ref{eqn-feb28121}) if $m \equiv 1 \pmod{8}$, and only the term
$x^{2^{2h}}$ in the function of (\ref{eqn-feb28111}) cancels the monomial $x$ in the function in the right-hand
side of (\ref{eqn-feb28121}) if $m \equiv 5 \pmod{8}$.

The desired conclusions on the linear span and the minimal polynomial $\m_s(x)$ then follow from Lemmas \ref{lem-ls2},
\ref{lem-1Nf261},  and Equation
(\ref{eqn-1Niho}).
\end{proof}

The following theorem provides information on the code $\C_{s}$.

\begin{theorem}\label{thm-yue}
Let $m \ge 9$ be odd.
The binary code $\C_{s}$ defined by the sequence of (\ref{eqn-1Nihoseq}) has parameters
$[2^m-1, 2^{m}-1-\ls_s, d]$ and generator polynomial $\m_s(x)$,
where $\ls_s$ and $\m_s(x)$ are given in Lemma \ref{lem-1N2m1} and the minimum weight $d$ has the following
bounds:
\begin{eqnarray}\label{eqn-niho1b}
d \ge \left\{ \begin{array}{ll}
 2^{(m-1)/4} + 2 & \mbox{if } m \equiv 1 \pmod{8} \\
 2^{(m-1)/4}       & \mbox{if } m \equiv 5 \pmod{8}.
\end{array}
\right.
\end{eqnarray}
\end{theorem}

\begin{proof}
The dimension  and the generator polynomial of $\C_{s}$ follow from Lemma \ref{lem-1N2m1} and the
definition of the code $\C_s$.  We now derive the lower bounds on the minimum weight $d$. It is well
known that the codes generated by $\m_s(x)$ and its reciprocal have the same weight distribution. The
reciprocal of $\m_s(x)$ has the zeros $\alpha^{i+2^{2h}}$ for all $i$ in $\{0,1,2, \cdots, 2^h-1\}$ if
$m \equiv 1 \pmod{8}$, and for all $i$ in $\{1,2, \cdots, 2^h-1\}$ if $m \equiv 5 \pmod{8}$.
Note that $\C_s$ is an even-weight code.  Then the desired bounds on $d$ follow from the BCH bound.
\end{proof}

\begin{example}
Let $m=5$ and $\alpha$ be a generator of $\gf(2^m)^*$ with $\alpha^5 + \alpha^2 + 1=0$. Then
the generator polynomial of the code $\C_s$ is
$
\m_s(x)=x^6 + x^3 + x^2 + 1,
$
and $\C_s$ is a $[31, 25, 4]$ binary cyclic code and  optimal according to the Database.

\end{example}

\begin{example}
Let $m=9$ and $\alpha$ be a generator of $\gf(2^m)^*$ with $\alpha^9 + \alpha^4 + 1=0$. Then
the generator polynomial of the code $\C_s$ is
$ 
\m_s(x) =  x^{46} + x^{45} + x^{41} + x^{40} + x^{39} + x^{36} + x^{35} + x^{33} + x^{28} + x^{27} + x^{26} + x^{25} +  x^{24} + x^{22} + x^{21} + x^{20} + x^{19} + x^{14} + x^{12} + x^7 + x^4 +
    x^2 + x + 1 
$  
and $\C_s$ is a $[511, 465, d]$ binary cyclic code, where $d \ge 6$. The actual minimum weight
may be larger than 6.

\end{example}

\section{Binary cyclic codes from the monomials $f(x)=x^{2^{2h}-2^h+1}$, where $\gcd(m,h)=1$}\label{sec-Kasami}

Define $f(x)=x^e$, where $e=2^{2h}-2^h+1$ and $\gcd(m,h)=1$.
In this section, we have the following additional restrictions on $h$:
\begin{eqnarray}\label{eqn-Hcondition}
1 \le h \le \left\{  \begin{array}{l}
                   \frac{m-1}{4} \mbox{ if } m \equiv 1 \pmod{4}, \\
                   \frac{m-3}{4} \mbox{ if } m \equiv 3 \pmod{4}, \\
                   \frac{m-4}{4} \mbox{ if } m \equiv 0 \pmod{4}, \\
                   \frac{m-2}{4} \mbox{ if } m \equiv 2 \pmod{4}.
\end{array}
\right.
\end{eqnarray}

Note that
\begin{eqnarray}\label{eqn-Kasami}
\tr(f(x+1))
&=& \tr\left( (x+1) (x+1) ^{\sum_{i=0}^{h-1} 2^{h+i}}    \right)\\
&=& \tr\left( (x+1) \prod_{i=0}^{h-1} \left(x^{2^{h+i}}+1\right)     \right) \nonumber \\
&=& \tr\left(\sum_{i=0}^{2^h-1} x^{2^{h}i+1}  + \sum_{i=0}^{2^h-1} x^{i}     \right) \nonumber \\
&=& \tr\left(\sum_{i=0}^{2^h-1} x^{i+2^{m-h}} + \sum_{i=1}^{2^h-1} x^{i}     \right) +1.
\end{eqnarray}

The sequence $s^{\infty}$ of (\ref{eqn-sequence}) defined by $f(x)$ is then
\begin{eqnarray}\label{eqn-Kasamiseq}
s_t= \tr\left(\sum_{i=0}^{2^h-1} (\alpha^t)^{i+2^{m-h}}  + \sum_{i=1}^{2^h-1} (\alpha^t)^{i}     \right) +1
\end{eqnarray}
for all $t \ge 0$, where $\alpha$ is a generator of $\gf(2^m)^*$.

In this section, we deal with the code $\C_s$ defined by the sequence $s^{\infty}$ of
(\ref{eqn-Kasamiseq}).
It is noticed that the final expression of the function of (\ref{eqn-Kasami}) is of the same format
as that of the function of (\ref{eqn-1Niho}). The proofs of the lemmas and theorems in this section
are very similar to those of Section \ref{sec-2hminus1}. Hence, we present only the main results without
providing proofs.

We define the following two sets for convenience:
\begin{eqnarray*}
A=\{0,1,2, \cdots, 2^h-1\}, \ B=2^{m-h}+A=\{i+2^{m-h}: i \in A\}.
\end{eqnarray*}

\begin{lemma}\label{lem-K1Nf261}
Let $h$ satisfy the conditions of (\ref{eqn-Hcondition}).
For any $j \in B$, the size $\ell_j=|C_j|=m$.
\end{lemma}

\begin{proof}
The proof of Lemma \ref{lem-1Nf261} is easily modified into a proof for this lemma.
The detail is omitted.
\end{proof}

\begin{lemma}\label{lem-K1Nf262}
Let $h$ satisfy the conditions of (\ref{eqn-Hcondition}).
For any pair of distinct $i$ and $j$ in $B$,
$C_i \cap C_j = \emptyset$, i.e., they cannot be in the same $2$-cyclotomic
coset modulo $n$.
\end{lemma}

\begin{proof}
The proof of Lemma \ref{lem-1Nf262} is easily modified into a proof for this lemma.
The detail is omitted.
\end{proof}

\begin{lemma}\label{lem-Kfeb281}
Let $h$ satisfy the conditions of (\ref{eqn-Hcondition}).
For any $i+2^{m-h} \in B$ and odd $j \in A$,
\begin{eqnarray}
C_{i+2^{m-h}} \cap C_j = \left\{ \begin{array}{l}
                             C_j \mbox{ if } (i,j)=(0,1) \\
                             \emptyset \mbox{ otherwise.}
\end{array}
\right.
\end{eqnarray}
\end{lemma}

\begin{proof}
The proof of Lemma \ref{lem-feb281} is easily modified into a proof for this lemma.
The detail is omitted here.
\end{proof}

\begin{lemma}\label{lem-K1N2m1}
Let $h$ satisfy the conditions of (\ref{eqn-Hcondition}).
Let $s^{\infty}$ be the sequence of (\ref{eqn-Kasamiseq}). Then the linear span $\ls_s$ of $s^{\infty}$ is given by
\begin{eqnarray}\label{eqn-K1N2m0}
\ls_s =\left\{ \begin{array}{l}
                   \frac{m\left(2^{(h+2}+(-1)^{h-1}\right) +3}{3} \mbox{ if $h$ is even} \\
                   \frac{m\left(2^{h+2}+(-1)^{h-1}-6\right) +3}{3} \mbox{ if $h$ is odd.}
\end{array}
\right.
\end{eqnarray}
We have also
\begin{equation*}\label{eqn-K1N2m21}
\m_s(x) = (x-1) \prod_{i=0}^{2^{h}-1} m_{\alpha^{-i-2^{m-h} }}(x)
\prod_{1 \le 2j+1 \le 2^{h}-1 \atop \kappa_{2j+1}^{h} =1} m_{\alpha^{-2j-1}}(x)
\end{equation*}
if $h$ is even; and
\begin{equation*}\label{eqn-K1N2m31}
\m_s(x) =(x-1) \prod_{i=1}^{2^{h}-1} m_{\alpha^{-i-2^{m-h} }}(x)
\prod_{3 \le 2j+1 \le 2^{h}-1 \atop \kappa_{2j+1}^{h} =1} m_{\alpha^{-2j-1}}(x)
\end{equation*}
if $h$ is odd,
where  $\kappa_{2j+1}^{(h)}$ was
defined in Section \ref{sec-2hminus1}.
\end{lemma}

\begin{proof}
The proof of Lemma \ref{lem-1N2m1} is easily modified into a proof for this lemma.
The detail is omitted.
\end{proof}

The following theorem provides information on the code $\C_{s}$.

\begin{theorem} \label{thm-Kyue}
Let $h$ satisfy the conditions of (\ref{eqn-Hcondition}).
The binary code $\C_{s}$ defined by the sequence of (\ref{eqn-Kasamiseq}) has parameters
$[2^m-1, 2^{m}-1-\ls_s, d]$ and generator polynomial $\m_s(x)$,
where $\ls_s$ and $\m_s(x)$ are given in Lemma \ref{lem-K1N2m1} and the minimum weight $d$ has the following
bounds:
\begin{eqnarray}\label{eqn-Kniho1b}
d \ge \left\{ \begin{array}{ll}
 2^{h} + 2 & \mbox{if $h$ is even}  \\
 2^{h}       & \mbox{if $h$ is odd.}
\end{array}
\right.
\end{eqnarray}
\end{theorem}

\begin{proof}
The proof of Lemma \ref{thm-yue} is easily modified into a proof for this lemma with the helps of the
lemmas presented in this section.
The detail is omitted here.
\end{proof}

\begin{example}
Let $(m,h)=(5,2)$ and $\alpha$ be a generator of $\gf(2^m)^*$ with $\alpha^5 + \alpha^2 + 1=0$. Then
the generator polynomial of the code $\C_s$ is
$ 
\m_s(x)= x^{16} + x^{14} + x^{10} + x^{9} + x^8 + x^7 + x^5 + x^4 + x^3 + x^2 + x+ 1 
$ 
and $\C_s$ is a $[31, 15, 8]$ binary cyclic code. Its dual is a $[31,16,7]$ cyclic code. Both codes are optimal according to the Database.
In this example, the condition of (\ref{eqn-Hcondition}) is not satisfied. So the conclusions on the code of
this example do not agree with the conclusions of Theorem \ref{thm-Kyue}.
\end{example}

\begin{example}
Let $(m,h)=(7,2)$ and $\alpha$ be a generator of $\gf(2^m)^*$ with $\alpha^7 + \alpha + 1=0$. Then
the generator polynomial of the code $\C_s$ is
$ 
\m_s(x) =  x^{36} + x^{28} + x^{27} + x^{23} + x^{21} + x^{20} + x^{18} +
  x^{13} + x^{12} + x^9 + x^7 + x^6 + x^5  + 1
$  
and $\C_s$ is a $[127, 91, 8]$ binary cyclic code.

\end{example}

In this section, we obtained interesting results on the code $\C_s$ under the conditions
of (\ref{eqn-Hcondition}). When $h$ is outside the ranges, it may be hard to determine the
dimension of the code $\C_s$, let alone the minimum weight $d$ of the code. Hence,
it would be nice if the following open problem can be solved.

\begin{open}
Determine the dimension and the minimum weight of the code $\C_s$ of this section when $h$ satisfies
\begin{eqnarray}\label{eqn-Hcondition2}
\left\{  \begin{array}{l}
                 \frac{m-1}{2} \ge h >                  \frac{m-1}{4} \mbox{ if } m \equiv 1 \pmod{4}, \\
                 \frac{m-3}{2} \ge h >                   \frac{m-3}{4} \mbox{ if } m \equiv 3 \pmod{4}, \\
                 \frac{m-4}{2} \ge h >                   \frac{m-4}{4} \mbox{ if } m \equiv 0 \pmod{4}, \\
                 \frac{m-2}{2} \ge h >                   \frac{m-2}{4} \mbox{ if } m \equiv 2 \pmod{4}.
\end{array}
\right.
\end{eqnarray}
\end{open}

\section{Binary cyclic codes from a trinomial over $\gf(2^m)$}

In this section, we study the code $\mathcal{C}_s$ from the trinomial $x+x^{r}+x^{2^h-1}$ where
$\wt(r)=m-1$ and $0\leq h\leq \lceil {m\over 2}\rceil$. Before doing this, we first introduce some
notations and lemmas which will be used in the sequel. Let $\rho_i$ denote the  number of even integers in the 2-cyclotomic coset $C_i$. For each
$i\in \Gamma$, define
\begin{align}\label{eqn-def-v}
v_i={m\rho_i\over \ell_i} \bmod 2
\end{align}
where $\ell_i=|C_i|$.

\begin{lemma}\label{Lemma_inverse}\cite{SiDing}
With the notations as before,
\begin{align}\label{eqn-invese}
\tr((1+\alpha^t)^{2^m-2})=\sum_{j\in \Gamma}v_j\left(\sum_{i\in C_j}(\alpha^t)^i\right).
\end{align}
Furthermore, the total number of nonzero coefficients of $\alpha^{it}$ in (\ref{eqn-invese})
is equal to $2^{m-1}$.
\end{lemma}

\begin{lemma}\label{Lemma_thirdclass}
Let $m\geq 4$, $r$ be an integer with $1\leq r\leq 2^m-2$ and $\wt(r)=m-1$, and $h$ an integer with $0\leq h\leq \lceil {m\over 2}\rceil$.
Let $s^{\infty}$ be the sequence of (\ref{eqn-sequence}), where
$
f(x)=x+x^{r}+x^{2^h-1}.
$ 
Then
the linear span of $s^{\infty}$ is given by
\begin{align}\label{eqn-linspan-thirdclass}
\ls_s=\left\{
\begin{array}{ll}
2^{m-1}+m, &\textrm{if~} m\textrm{~is~odd~and~}h=0 \\
2^{m-1}-m, &\textrm{if~} m\textrm{~is~even~and~}h=0 \\
2^{m-1}, &\textrm{if~} h\neq 0
\end{array} \right.\ \
\end{align}
and the minimal polynomial of $s^{\infty}$ is given by
\begin{align}\label{eqn-minimalpoly-thirdclass}
\m_s(x)=\left\{
\begin{array}{ll}
m_{\alpha^{-1}}(x)\prod_{j\in \Gamma\setminus \{1\},v_j =1}m_{\alpha^{-j}}(x), &\textrm{if~} m\textrm{~is ~odd~and~}h=0 \\
\prod_{j\in \Gamma\setminus \{1\},v_j =1}m_{\alpha^{-j}}(x), &\textrm{if~} m\textrm{~is ~even~and~}h=0 \\
\prod_{j\in \Gamma,u_j =1}m_{\alpha^{-j}}(x), &\textrm{if~} h\neq 0
\end{array} \right.\ \
\end{align}
where $m_{\alpha^{-j}}(x)$ is the minimal polynomial of $\alpha^{-j}$ over $\gf(2)$, $u_1=(v_1+\kappa^{(h)}_1+1) \bmod{2}$,  $u_{2j+1}=(v_{2j+1}+\kappa^{(h)}_{2j+1}) \bmod{2}$ for $3\leq 2j+1\leq 2^h-1$, and $u_j=v_j$
for~ $j\in \Gamma\setminus \{2i+1: 1\leq 2i+1\leq 2^h-1\}$. Herein, $\kappa^{(h)}_{2i+1}$ is given by (\ref{eqn-def-kappa}).
\end{lemma}

\vspace{2mm}

\begin{proof}
Note that $\wt(r)=\wt(2^{m}-2)=m-1$. It then follows from Lemma \ref{Lemma_equivalent} and properties of the trace function that
$ 
\tr(x^r)=\tr(x^{2^m-2}) \textrm{~~for~all~~} x\in \gf(2^m).
$

We first deal with the case of $h=0$ where  $f(x)=1+x+x^{r}$. According to Lemma \ref{Lemma_inverse}, one has
\begin{align}\label{eqn-s-class3_1}
\tr(f(1+\alpha^t))&=\tr(1)+\tr(1+\alpha^t)+\tr((1+\alpha^t)^{2^{m-2}-2})\nonumber \\
&=\sum_{i\in C_1}(\alpha^t)^i+\sum_{j\in \Gamma}v_j\left(\sum_{i\in C_j}(\alpha^t)^i\right)\nonumber\\
&=\sum_{i\in C_1}(1+v_1)(\alpha^t)^i+\sum_{j\in \Gamma\setminus \{1\}}v_j\left(\sum_{i\in C_j}(\alpha^t)^i\right)
\end{align}
where $1+v_1$ is performed modulo 2. By the definition of $v_i$, $v_1=(m-1) \bmod 2$. The desired conclusion on the linear span and minimal polynomial of $s^{\infty}$ for the case $h=0$ then follows
from Equation (\ref{eqn-s-class3_1}) and Lemma \ref{lem-ls2}.

Now we assume that $h\neq 0$. From the proof of Lemma \ref{lem-22mm1}, we know that
\begin{equation*}
\tr\left(\sum_{i=1}^{2^h-1}x^i\right)=\sum_{\scriptstyle 1\leq 2j+1\leq 2^h-1 \atop \scriptstyle \kappa^{(h)}_{2j+1}=1}\tr(x^{2j+1})
\end{equation*}
It then follows from Lemma \ref{Lemma_inverse} that
\begin{align}\label{eqn-s-general}
\tr(f(1+\alpha^t))&=\tr(1+\alpha^t)+\tr((1+\alpha^t)^{2^{m-2}-2})+\tr((1+\alpha^t)^{2^h-1})\nonumber \\
&=\sum_{i\in C_1}(\alpha^t)^i+\sum_{j\in \Gamma}v_j\left(\sum_{i\in C_j}(\alpha^t)^i\right)+\sum_{j\in \Gamma_1}\kappa^{(h)}_j\left(\sum_{i\in C_j}(\alpha^t)^i\right)\nonumber\\
&=\sum_{j\in \Gamma}u_j\left(\sum_{i\in C_j}(\alpha^t)^i\right)
\end{align}
where $\Gamma_1=\{2i+1: 1\leq 2i+1\leq 2^h-1\}$, $u_1=(v_1+\kappa^{(h)}_1+1) \bmod{2}$,
$u_j=(v_j+\kappa^{(h)}_j) \bmod{2}$ for $j\in \Gamma_1 \setminus \{1\}$ and $u_j=v_j$
for $j\in \Gamma \setminus \Gamma_1$.
The minimal polynomial in (\ref{eqn-minimalpoly-thirdclass}) then follows from  Equation (\ref{eqn-s-general}).

Finally, we show  that the linear span  of $s^{\infty}$ is equal to $2^{m-1}$ when $h\neq 0$, i.e.,
the total number of nonzero coefficient of $\alpha^{it}$ in (\ref{eqn-s-general}) is equal to $2^{m-1}$.
According to Lemmas \ref{Lemma_inverse} and \ref{Lemma_coset},
it is sufficient to prove that the number of $j\in \Gamma_1$ such that $u_j\neq 0$ is equal to the number
of $j\in \Gamma_1$ such that $v_j\neq 0$. By the definition of $v_j$ and Lemma \ref{lemma_odd_number}, we have
$v_1=(m-1) \bmod{2}$, $v_3=(m-2) \bmod{2}$.  According to the definition of $\kappa^{(h)}_j$ in (\ref{eqn-def-kappa}), we have $\kappa^{(h)}_1=h \bmod{2}$ and $\kappa^{(h)}_3=(h-1) \bmod{2}$.
Thus, $u_1=(m+h) \bmod{2}$ and $u_3=(m+h-1) \bmod{2}$. Thus, whenever $m$ is even or odd, there are exactly one nonzero $u_i$ and $v_i$ for $i\in \{1,3\}$.
Note that, for any integer $a$ with $2\leq a\leq h-1$, the number of odd integers $j$ satisfying
$2^a<j<2^{a+1}$ is equal to $2^{a-1}$. It is clear that
when $2j+1$ rangers over the odd integers between $2^a$ and $2^{a+1}$, the number of $2j+1$
such that $\wt(2j+1)$ to be even is equal to $2^{a-2}$. It then follows from Lemma \ref{lemma_odd_number} that
\begin{align*}
|\{2^a<2j+1<2^{a+1}: v_{2j+1}=1\}|=|\{2^a<2j+1<2^{a+1}: v_{2j+1}=0\}|=2^{a-2}.
\end{align*}
On the other hand, when $2j+1$ runs over the odd integers from $2^a$ to $2^{a+1}$,
$\kappa^{(h)}_{2j+1}$ has the same value for these $j$'s due to the definition of $\kappa^{(h)}_{2j+1}$ in (\ref{eqn-def-kappa}). If $\kappa^{(h)}_j=0$, $u_j=v_j$, and otherwise $u_j=v_j+1$.
Thus we have
\begin{align*}
|\{2^a<2j+1<2^{a+1}: u_{2j+1}=1\}|=|\{2^a<2j+1<2^{a+1}: u_{2j+1}=0\}|=2^{a-2}.
\end{align*}
It then follows that 
\begin{align*}
|\{j\in \Gamma_1: u_j=1\}|=|\{j\in \Gamma_1: v_j=1\}|.
\end{align*}
By Lemma \ref{Lemma_coset}, $|C_j|=m$ for each $j\in \Gamma_1$. The conclusion then follows 
from the analysis above. 
\end{proof}

\vspace{2mm}
\begin{theorem}\label{Theorem_third}
The code $\mathcal{C}_s$ defined by the sequence of Lemma \ref{Lemma_thirdclass}
has parameters $[n,n-\ls_s,d]$ and generator polynomial $\m_s(x)$ of
(\ref{eqn-minimalpoly-thirdclass}), where $\ls_s$ is given by (\ref{eqn-linspan-thirdclass}) and
\begin{align*}
d\geq \left\{
\begin{array}{ll}
8, &\textrm{if~} m\textrm{~is ~odd~and~}h=0 \\
3, &\textrm{if~} m\textrm{~is ~even~and~}h=0. \\
\end{array} \right.\ \
\end{align*}
\end{theorem}

\begin{proof}
The dimension of $\C_s$ follows from Lemma \ref{Lemma_thirdclass} and the definition
of this code. We only need to prove the conclusion on the minimal distance $d$ of $\C_s$. It is known that
codes generated by any polynomial $g(x)$ and its reciprocal have the same weight distribution. When $m$ is odd and $h=0$,
since $v_3=v_5=1$, the reciprocal of $\M_s(x)$ has zeros $\alpha^t$ for $t\in \{0,1,2,3,4,5,6\}$. It then follows from the BCH bound
that $d\geq 8$. When $m$ is odd and $h=0$, note that $v_7=v_{13}=1$, the reciprocal of $\M_s(x)$ has zeros $\alpha^t$ for $t\in \{13,14\}$.
By the BCH bound, $d\geq 3$.
\end{proof}

\begin{open}
Develop a tight lower bound
on the minimal distance of the code $\mathcal{C}_s$ in Theorem
\ref{Theorem_third} for the case that $h>0$.
\end{open}

\vspace{2mm}

\begin{remark}
When $r=2^m-2$ and $h=1$, $f(x)$ becomes the monomial $x^{2^m-2}$ which is the the inverse APN function.
It was pointed in \cite{Ding56} that the code $\mathcal{C}_s$ from the inverse APN function may have poor minimal 
distance when $m$ is even. However, when we choose some other $h$, the corresponding codes may have excellent 
minimal distance. This is demonstrated by some of the examples below. 
\end{remark}

\begin{example}
Let $m=4$, $r=2^m-2$, $h=1$, and $\alpha$ be the generator of $\gf(2^m)$ with $\alpha^4+\alpha+1=0$.
Then the generator polynomial of $\C_s$  is
$ 
\M_s(x)=x^8 + x^7 + x^5 + x^4 + x^3 + x + 1 
$ 
and $\C_s$ is a $[15, 7, 3]$ binary cyclic code. It is not optimal.
\end{example}

\begin{example}
Let $m=4$, $r=2^m-2$, $h=0$, and $\alpha$ be the generator of $\gf(2^m)$ with $\alpha^4+\alpha+1=0$.
Then the generator polynomial of $\C_s$  is
$ 
\M_s(x)=x^4+x+1
$
and $\C_s$ is a $[15, 11, 3]$ optimal binary cyclic code. The optimal binary linear code with the
same parameters in the Database is not cyclic.
\end{example}

\begin{example}
Let $m=4$, $r=2^m-2$, $h=2$, and $\alpha$ be the generator of $\gf(2^m)$ with $\alpha^4+\alpha+1=0$.
Then the generator polynomial of $\C_s$  is
$
\M_s(x)=x^8+x^7+x^6+x^4+1 
$ 
and $\C_s$ is a $[15, 7, 5]$ optimal binary cyclic code. The optimal binary linear code with the
same parameters in the Database is not cyclic.
\end{example}

\begin{example}
Let $m=5$, $r=2^m-2$, $h=0$, and $\alpha$ be the generator of $\gf(2^m)$ with $\alpha^5+\alpha^2+1=0$.
Then the generator polynomial of $\C_s$  is
$ 
\M_s(x)=x^{21} + x^{18} + x^{17} + x^{15} + x^{13} + x^{10} + x^5 + x^4 + x^3 + x^2 + x + 1
$
and $\C_s$ is a $[31,10,12]$ optimal binary cyclic code.
\end{example}

\begin{example}
Let $m=5$, $r=2^m-2$, $h=1$, and $\alpha$ be the generator of $\gf(2^m)$ with $\alpha^5+\alpha^2+1=0$.
Then the generator polynomial of $\C_s$  is
$ 
\M_s(x)=x^{16} + x^{14} + x^{13} + x^{10} + x^9 + x^8 + x^7 + x^6 + x^5 + x^2 + x + 1
$ 
and $\C_s$ is a $[31,15,8]$ optimal binary cyclic code. The optimal binary linear code with the
same parameters in the Database is not cyclic.
\end{example}

\section{Open problems regarding binary cyclic codes from monomials}

In the previous sections, we investigated binary cyclic codes defined by some monomials. 
It would be good if the following open problems could be solved.

\begin{open}
Determine the dimension and the minimum weight of the code $\C_s$ defined by the monomials
$x^e$, where $e=2^{(m-1)/2}+2^{(3m-1)/4}-1$ and $m \equiv 3 \pmod{4}$.
\end{open}

\begin{open}
Determine the dimension and the minimum weight of the code $\C_s$ defined by the monomials
$x^e$, where $e=2^{4i}+2^{3i}+2^{2i}+2^i-1$ and $m=5i$.
\end{open}

\section{Concluding remarks and summary}

In this paper, we constructed a number of families of cyclic codes with monomials and trinomials
of special types. The dimension of some of the codes is flexible. We determined the minimum
weight for some families of cyclic codes, and developed tight lower bounds for other
families of cyclic codes. The main results of this paper showed that the  approach of
constructing cyclic codes with polynomials is promising. While it is rare to see optimal
cyclic codes constructed with tools in algebraic geometry and algebraic function fields,
the simple construction of cyclic codes with monomials and trinomials over $\gf(r)$ employed
in this paper is impressive in the sense that it has produced optimal and almost
optimal cyclic codes. 

The binary sequences defined by some of the monomials and trinomials have large linear
span. These sequences have also reasonable autocorrelation property. They could be
employed in certain stream ciphers as keystreams. So the contribution of this paper in cryptography
is the computation of the linear spans of these sequences.

It is known that long BCH codes are bad \cite{LinWeldon}. However, it was indicated in
\cite{BJ74,MartW} that there may be good cyclic codes. The cyclic codes presented in
this paper proved that some families of cyclic codes are in fact very good.

Four open problems regarding binary cyclic codes were proposed in this paper. The reader is cordially
invited to attack them.

\end{document}